\documentclass[11pt,a4paper]{article}
\usepackage{authblk}%authors package
\usepackage[left=3cm,right=3cm,top=3cm,bottom=3cm]{geometry} % page settings
\usepackage{amsmath,amsthm,amssymb,epsfig,latexsym,color} % provides many mathematical environments & tools
\usepackage{indentfirst}
\usepackage{lipsum}

\usepackage[
    backend=bibtex, %backend tells biblatex what you will be using to process the bibliography file
    natbib=true,
    style=phys,
    biblabel=brackets,
    pageranges=false,
   minalphanames=3,
   maxalphanames=4,
    citestyle=alphabetic,
    sorting=anyt,
    giveninits=true,
    abbreviate=true,
    doi=false,
    url=false,
    isbn=false,
    block=space,
    backref=false,%cit on pag x
    backrefstyle=two,
]{biblatex}

\DeclareFieldFormat{titlecase}{#1}

\DeclareFieldFormat{bibentrysetcount}{\mkbibparens{\mknumalph{#1}}}
\DeclareFieldFormat{labelalphawidth}{\mkbibbrackets{#1}}
\DeclareFieldFormat{shorthandwidth}{\mkbibbrackets{#1}}

\defbibenvironment{bibliography}
  {\list
     {\printtext[labelalphawidth]{%
        \printfield{labelprefix}%
        \printfield{labelalpha}%
        \printfield{extraalpha}}}
     {\setlength{\labelwidth}{\labelalphawidth}%
      \setlength{\leftmargin}{\labelwidth}%
      \setlength{\labelsep}{\biblabelsep}%
      \addtolength{\leftmargin}{\labelsep}%
      \setlength{\itemsep}{\bibitemsep}%
      \setlength{\parsep}{\bibparsep}}%
      }
  {\endlist}
  {\item}

\usepackage{caption}%subfigures
\usepackage{subcaption}%subfigures

\usepackage{tikz}
    \usetikzlibrary{decorations.pathreplacing}
    \usetikzlibrary{shapes.geometric}
    \usetikzlibrary{plotmarks}
    \usetikzlibrary{calc}
    \usetikzlibrary{fadings}
\usepackage{pgfplots}
\usepgfplotslibrary{colormaps}
\usepgfplotslibrary{patchplots}

\pgfplotsset{
	compat=newest
}

\usepackage[bookmarks=true,colorlinks,linkcolor=blue,urlcolor=blue,citecolor=blue,breaklinks]{hyperref}

\newcommand{\so}{\scriptscriptstyle \rm I}
\newcommand{\st}{\scriptscriptstyle \rm I\hspace{-1pt}I}

\newcommand{\qo}{{\rm i}}
\newcommand{\qt}{{\rm ii}}

\newcommand{\be}[1]{\begin{equation}\label{#1}}
\newcommand{\ee}{\end{equation}}
\newcommand{\bu}{\bar u}
\newcommand{\bv}{\bar v}

\newcommand{\rank}{\mathop{\rm rank}}

\newcommand{\tr}{\mathop{\rm tr}}

\newcommand\CC{\mathbb C}

\newcommand{\ben}{\begin{eqnarray}}
\newcommand{\een}{\end{eqnarray}}

\newtheorem{theorem}{Theorem}[section]

\newtheorem{prop}{Proposition}[section]

\newtheorem{Def}{Definition}[section]
\def\<{\langle}
\def\>{\rangle}

\numberwithin{equation}{section}

\begin{document}

\date{}

\title{Modified rational six vertex model on the rectangular lattice}

\author[1]{S. Belliard \thanks{samuel.belliard@lmpt.univ-tours.fr}}
\author[2]{R.A. Pimenta \thanks{rodrigo.alvespimenta@umanitoba.ca}}
\author[3]{N.A. Slavnov \thanks{nslavnov@mi-ras.ru}}

\affil[1]{Institut Denis-Poisson, Universit\'e de Tours, Universit\'e d'Orl\'eans, Parc de Grammont, 37200 Tours, France}
\affil[2]{Department of Physics and Astronomy, University of Manitoba, Winnipeg R3T 2N2, Canada}
\affil[3]{Steklov Mathematical Institute of Russian Academy of Sciences, 8 Gubkina str., Moscow, 119991,  Russia}

\affil[ ]{\textit {}}
\maketitle

\begin{abstract}
We consider a rational six vertex model on a rectangular lattice with boundary conditions that generalize the usual domain wall type. We find that the partition function of the inhomogeneous version of this model is given by a modified Izergin determinant. The proofs are based on the quantum inverse scattering method and its representation theory together with elementary linear algebra.
\end{abstract}

\section{Introduction}

A partition function is the central object in Statistical Mechanics. Computing the partition function is often a difficult combinatorial problem that sometimes
can be solved exactly, particularly in two dimensional lattices \cite{Onsager1944,Gaudin83,Baxter1985}. A prominent
example is the six vertex model, which can be defined on a rectangular lattice
by assigning two possible states to each of its edges (see Figure \ref{fig1}). The possible configurations
around a vertex are constrained by the ice rule. Boundary conditions must be imposed, and it turns out
that they play an
important role in the nature of the mathematical representation of the partition function.

Different solvable boundary conditions have been considered for
the six vertex model. The most common are the torus, the cylinder and
fixed boundary conditions like domain wall and reflecting end, as well as mixtures of them. The first two
types can be computed via the diagonalization of an appropriate transfer matrix,
which can in principle be done
using Bethe ansatz, as early noted by Lieb
in his solution of the square ice \cite{Lieb1967}. On the other hand, fixed
boundary conditions like domain wall type admit a determinant representation \cite{Korepin1982,Izergin1987}, so do the
reflecting end case \cite{Tsuchiya1998}. In this paper, we will focus on fixed boundary conditions
that generalize the domain wall type.

The domain wall boundary condition for the six vertex model in the square lattice
was introduced by Korepin \cite{Korepin1982},
where a recurrence relation for the partition function was found. The solution
of the recurrence was later given by Izergin \cite{Izergin1987} in the form of a determinant
(see, for example, the monograph \cite{KBI} for details). The partition
function on the rectangular lattice, also called a partial domain wall boundary, was considered
recently and it is also given by a determinant \cite{Foda2012}.

Here we propose a generalization of the rational six vertex partition function
on the rectangular lattice with four
arbitrary boundaries.  In the language of the
quantum inverse scattering method, this partition function
is given by certain expectation values of a string of modified creation operators,
which arise in the context of the modified algebraic
Bethe ansatz (see for instance \cite{BePi15,BSV18,BSV18a,BS19a,BSV21}
and references therein). The expectation values are associated with four arbitrary
wall states labeled with the compass directions (\ref{compass}).
We argue that the partition function satisfies a homogeneous
system of linear equations, which follows
directly from the (modified) triangular representation theory of the Yangian  $Y(gl(2))$ and
Yang-Baxter algebra via certain off-shell relations. We then show that the linear system is solved by a modified Izergin determinant (\ref{ZmodIKmn}).

Let us recall that the partition function of the six vertex model
with domain wall boundary condition was shown to
solve a system of functional equations \cite{Galleas2010} (see
also \cite{Stroganov2006}), as well
as a system of algebraic equations \cite{Minin2023}.

The linear system approach was recently discovered
in the computation of scalar products between on-shell and off-shell
Bethe states \cite{BS19}. Its powerfulness has been demonstrated in the
computation of on-shell/off-shell scalar products of Bethe states in the
closed XYZ chain \cite{SZZ20} and in the
open XXZ chain with general integrable boundaries \cite{BPS21}. Here we
show that it can also be used to compute partition functions opening a
new avenue of possibilities.

This paper is organized as follows. In Section \ref{sec:2}, we
recall the definition of an arbitrary vertex model on the rectangular lattice, its partition
function with arbitrary twists and some essential ingredients of the quantum inverse scattering
method. In Section \ref{sec:3}, auxiliary operators and the associated
representation theory are studied in the rational six vertex case. Next, in Section \ref{sec:4}, we derive
the homogeneous linear systems satisfied by the partition function
and construct the solution in the form of a modified Izergin determinant.
Our concluding remarks and further directions of research are presented
in Section \ref{sec:conc}. Some technical details are given in the appendices
\ref{sec:yangian} and \ref{sec:mi}.

\vspace{1cm}

\textbf{Notation.} We use a shorthand notation for sets of variables and products over them.
For example, we denote a set of $m$ variables $u_j$ by $\bar u = \{u_1,\dots,u_m\}$. We
usually leave
the cardinality implicit and note it as $\# \bar u=m$. The removal of the $i$-th element of the set
$\bar u$ is denoted $\bar u_i=\bar u \backslash u_i$. For the product of
a two variable function $g(u,v)$ over the set $\bar u$ we use,
\ben\label{nota1}
g(z,\bar u) = \prod_{x\in \bar u}g(z,x),\quad
g(\bar u,z) = \prod_{x\in \bar u}g(x,z),\quad
g(\bar u,\bar v) = \prod_{x\in \bar u,y\in \bar v}g(x,y)\,.
\een
We will also use such notation for the product of commuting operators
\ben\label{nota1o}
B(\bar u)= \prod_{x\in \bar u} B(x).
\een
If no product is involved, a vertical bar is used to indicate the
multivariable function dependency, {\it e.g.},
\ben\label{nota2}
s(u|\bar u) = s(u|{u_1,\dots,u_{m}})\,,\quad r(\bar u|\bar w) = r(u_1,\dots, u_m|v_1,\dots,v_n).
\een
The following functions will be used,
\ben\label{funcs}
g(u,v)=\frac{c}{u-v},\quad f(u,v)=\frac{u-v+c}{u-v},\quad h(u,v)=\frac{u-v+c}{c},
\quad \tilde h(u,v)=\frac{u-v-c}{c}\,.
\een
They satisfy the relations
\ben\label{prop-funcs}
f(u,v)=g(u,v)+1=g(u,v)h(u,v)=g(v,u)\tilde h(v,u).
\een

\section{Partition function and quantum groups}\label{sec:2}

In this section we review some basic concepts in the theory of integrable vertex models,
including
the operator formulation of the partition function with general open boundary condition.
The material in this section is valid for an arbitrary
integrable vertex model
with arbitrary number of states.

\subsection{R-matrix and Yang-Baxter equation}\label{sec:yangbaxter}

%%%%%%%%%%%%%%%%%%%%%%%%%%%%%%%%%%%%%%%%%%%%%%%%%%%%%%%%%%%%%%%%%%%%%

Let us recall the inhomogeneous vertex model on the $m \times n$ rectangular lattice from the quantum group formalism (see Figure \ref{fig1}).
A finite dimensional vector space  denoted $V_{a_i}$ carrying a free parameter $u_i$,
also called inhomogeneity, is associated with each line of the lattice $i=1,\dots, m$. Similarly, a
finite dimensional vector
space $V_{b_i}$ with parameter  $v_i$  is associated with each column $i=1,\dots, n$. Every vertex is labeled by a pair of variables $(u, v)$ and it is encoded by an invertible matrix  $R_{ab}(u,v)$ (see Figure \ref{F-R-matrix}).
The entries of this matrix are the statistical weights of the model. Explicitly, we have,
\ben
R_{ab}(u,v)=\sum_{jk,\ell m}R_{jk,\ell m}(u,v) (E_{jk})_a \otimes (E_{\ell m})_b
\een
where the $(E_{ij})_a$ are unit matrices (that satisfy $E_{ij}E_{kl}=\delta_{jk}E_{il}$) acting on the space $V_a$, and the sum is taken over $(j,k)\in \{1,\dots, \dim(V_a)\}^2$ and $(\ell,m)\in \{1,\dots, \dim(V_b)\}^2$. One imposes that this matrix is a solution of the Yang Baxter equation,
\ben\label{YB}
R_{ab}(u,v)R_{ac}(u,w)R_{bc}(v,w)=R_{bc}(v,w)R_{ac}(u,w)R_{ab}(u,v)\,,
\een
and therefore ensures the integrability of the model.
%
%%%%%%%%%%%%%%%%%%%%%%%%%%%%%%%%% Figure-1 %%%%%%%%%%%%%%%%%%%%%%%%%%%%%%%
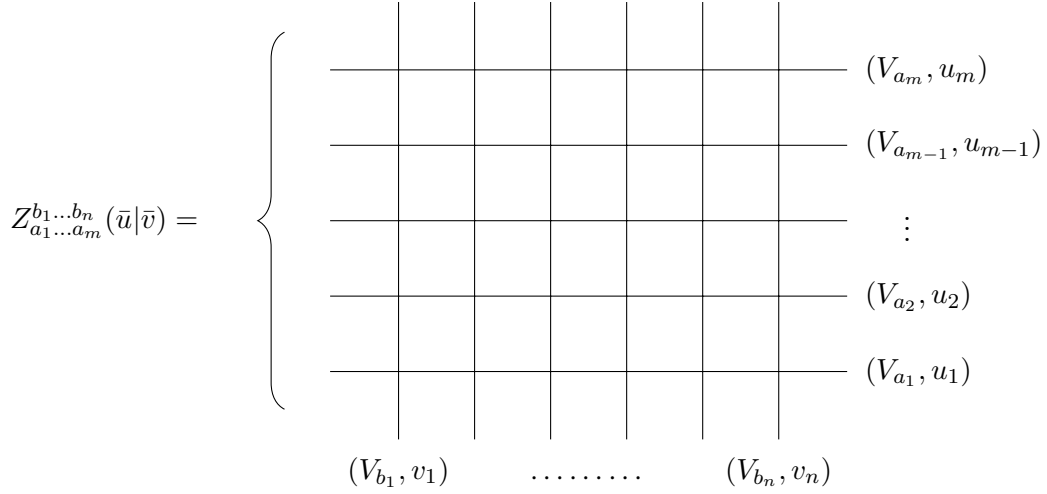
\begin{figure}[h!]
\centering
\begin{tikzpicture}
\draw (0.1,0.1) grid (6.9,5.9);
%\foreach \i in {1,...,5}
%    \foreach \j in {1,...,5}
%        \fill (\i,\j) circle (2pt);
\fill (7,1) node[right] {$(V_{a_1},u_{1})$};
\fill (7,2) node[right] {$(V_{a_2},u_{2})$};
\fill (8-0.5,3) node[right] {$\vdots$};
\fill (7,4) node[right] {$(V_{a_{m-1}},u_{m-1})$};
\fill (7,5) node[right] {$(V_{a_m},u_{m})$};
\fill (-1.5,3) node[left] {$Z^{b_1\dots b_n}_{a_1\dots a_m}(\bar u|\bar v)=$ };
\draw [decorate,decoration={brace,amplitude=10pt}]
(-0.5,0.5) -- (-0.5,5.5) ;
\fill (1,0) node[below] {$(V_{b_1},v_{1})$};
\fill (3.5,-0.2) node[below] {$\cdots$$\cdots$$\cdots$};
\fill (6,0) node[below] {$(V_{b_{n}},v_{n})$};
\end{tikzpicture}
\caption{\label{fig1} Inhomogeneous vertex model on the $m \times n$ rectangular lattice.
}
\end{figure}

\begin{figure}
\centering
\begin{tikzpicture}
\draw (0.1,0.1) grid (1.9,1.9);
\fill (0,1) node[left] {$(V_{a},u)$};
\fill (-1.5,1) node[left] {$R_{ab}(u,v)=$ };
\fill (1,0) node[below] {$(V_{b},v)$};
\end{tikzpicture}
\caption{\label{F-R-matrix} Graphical picture of a vertex and the R-matrix.}
\end{figure}
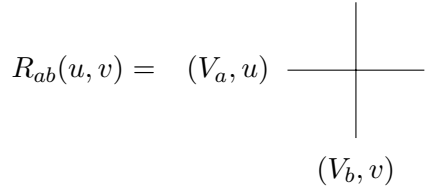
In the so-called auxiliary space formalism, we can write the matrix of the partition function in the form,
\ben\label{Zmatrix}
Z^{b_1\dots b_n}_{a_1\dots a_m}=\overrightarrow{\prod_{i=1}^m}\overrightarrow{\prod_{j=1}^n}R_{a_ib_j}(u_i,v_j)\,,
\een
which is an endomorphism of $V_{a_1}\otimes \dots  \otimes V_{a_m} \otimes V_{b_1}\otimes \dots  \otimes V_{b_n}$.
The following exchange relations follow from the Yang-Baxter algebra,
\ben\label{yba1}
R_{a_ia_{i+1}}(u_i,u_{i+1})Z^{b_1\dots b_n}_{a_1\dots a_ia_{i+1} \dots a_m}= Z^{b_1\dots  b_n}_{a_1\dots a_{i+1}a_i \dots a_m}R_{a_ia_{i+1}}(u_i,u_{i+1})
\een
and
\ben\label{yba2}
R_{b_ib_{i+1}}(v_i,v_{i+1})Z^{b_1\dots b_ib_{i+1}\dots  b_n}_{a_1 \dots a_m}= Z^{b_1\dots b_{i+1}b_i \dots  b_n}_{a_1\dots a_m}R_{b_ib_{i+1}}(v_i,v_{i+1}).
\een

We now assume that the R-matrix has the properties,
\ben\label{inv}
[R_{ab}(u,v),B_aB_b]=[R_{ab}(u,v),\Hat B_a \Hat B_b]=0\,,
\een
where $B$ and $\Hat B$ are matrices in $\text{End}(V)$. Using the
matrices $B$ and $\Hat B$ we can add a twist to each auxiliary space. This motivates
the definition of a ``quasiperiodic", or twisted inhomogeneous partition function, by taking
the trace over each space, namely,
\ben\label{Znm}
Z_{mn}(\bar u|\bar v|B|\overline B)=\tr_{\bar a,\bar b}\Big((\prod_{i=1}^n\hat B_{b_i})(\prod_{j=1}^mB_{a_j}) Z^{b_1\dots b_n}_{a_1\dots a_m}(\bar u|\bar v)\Big)
\een
where we used the shorthand notation (\ref{nota1},\ref{nota2}).
This definition generalizes the standard definition of the periodic partition function given by $B_a=\hat B_a=1_{V_a}$, that always satisfies (\ref{inv}) (see {\it e.g.} Chapter 7.5 of \cite{Chari}).
It is easy to show that (\ref{Znm}) is a symmetric function over the sets $\bar u$ or $\bar v$.
Moreover, if we relax the symmetric property for one set of variables, we can introduce different
twist matrices for each vector space. That is, the transformation $\prod_{i=1}^n \Hat B_{b_i} \to \prod_{i=1}^n\Hat B^{(i)}_{b_i}$ breaks the symmetry over permutation of the set  $\bar v$ and the transformation $\prod_{j=1}^m B_{a_j} \to \prod_{j=1}^m B^{(j)}_{a_j}$ break the symmetry over permutation of the set  $\bar u$.

\subsection{Operator formalism and quantum group}\label{sec:OFQG}

We now express the matrix of the partition function (\ref{Zmatrix})
in terms of products of single row monodromy matrices, either in the horizontal or vertical
direction of the lattice.

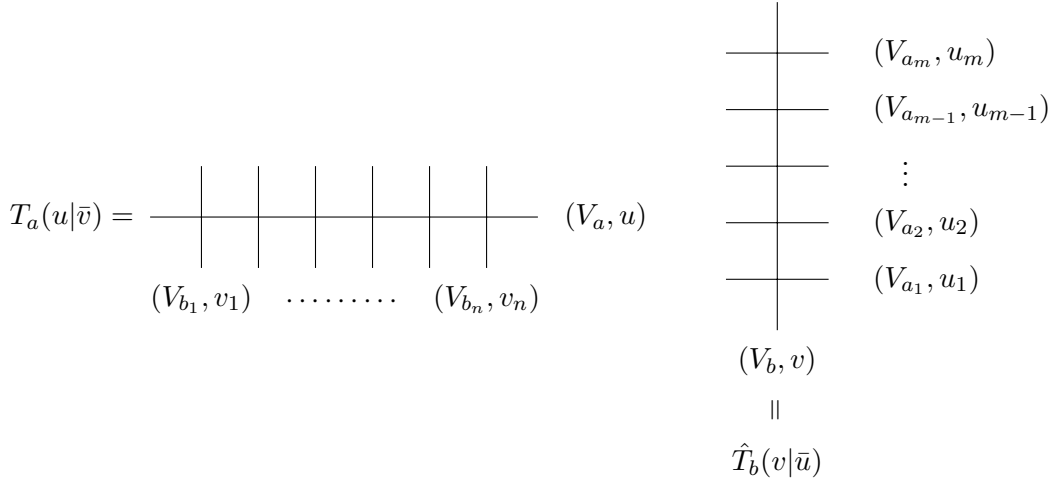
\begin{figure}[h!]
\begin{subfigure}[c]{0.50\textwidth}
\centering
\begin{tikzpicture}[scale=0.75]
\draw (0.1,0.1) grid (6.9,1.9);
\fill (9,1) node[left] {$(V_{a},u)$};
\fill (0,1) node[left] {$T_a(u|\bar v)=$ };
\fill (1,0) node[below] {$(V_{b_1},v_{1})$};
\fill (3.5,-0.2) node[below] {$\cdots$$\cdots$$\cdots$};
\fill (6,0) node[below] {$(V_{b_{n}},v_{n})$};
\end{tikzpicture}
%    \caption{}
  \end{subfigure}
  \quad
  \begin{subfigure}[c]{0.50\textwidth}
\centering
\begin{tikzpicture}[scale=0.75]
\draw (0.1,0.1) grid (1.9,5.9);
\fill (2.5,1) node[right] {$(V_{a_1},u_{1})$};
\fill (2.5,2) node[right] {$(V_{a_2},u_{2})$};
\fill (3,3) node[right] {$\vdots$};
\fill (2.5,4) node[right] {$(V_{a_{m-1}},u_{m-1})$};
\fill (2.5,5) node[right]  {$(V_{a_m},u_{m})$};
\fill (1,0) node[below] {$(V_{b},v)$};
\fill (1,-1.3) node  {\rotatebox{90}{$=$}};
\fill (1,-2.2) node {$\Hat T_{b}(v|\bar u)$};
\end{tikzpicture}
%    \caption{}
  \end{subfigure}
\caption{\label{mono} Single row monodromy matrix (left)
and single column monodromy matrix (right).
}
\end{figure}

The horizontal single row monodromy matrix is given by
\ben\label{monoT}
T_a(u|\bar v)=\overrightarrow{\prod_{j=1}^n}R_{ab_j}(u,v_j),
\een
and can be represented as in the left panel of Figure \ref{mono}. Is is clear that (\ref{Zmatrix}) can be written
as,
\ben
Z^{b_1\dots b_n}_{a_1\dots a_m}(\bar u|\bar v)=\overrightarrow{\prod_{i=1}^m} T_{a_i}(u_i|\bar v).
\een
The single row monodromy matrix (\ref{monoT}) satisfies the RTT relation,
\ben\label{RTT}
R_{ab}(u_a,u_b)T_a(u_a|\bar v)T_b(u_b|\bar v)=T_b(u_b|\bar v)T_a(u_a|\bar v)R_{ab}(u_a,u_b)\,.
\een
It allows one to define the transfer matrix with an arbitrary twist $B_a$ given by
\ben\label{Bop}
B(u)=\tr_a(B_a T_a(u|\bar v)),
\een
which is represented on the left panel of Figure \ref{figBop}, and it is integrable since it leads to a family of mutually commuting operators,
\ben
[B(u_1),B(u_2)]=0.
\een
The partition function can then be written in the operator formulation,
\ben\label{ZtrA}
Z_{mn}(\bar u|\bar v|B|\overline B)=\tr_{\bar b} \Big((\prod_{i=1}^n\hat B_{b_i}) B(\bar u) \Big)\,.
\een

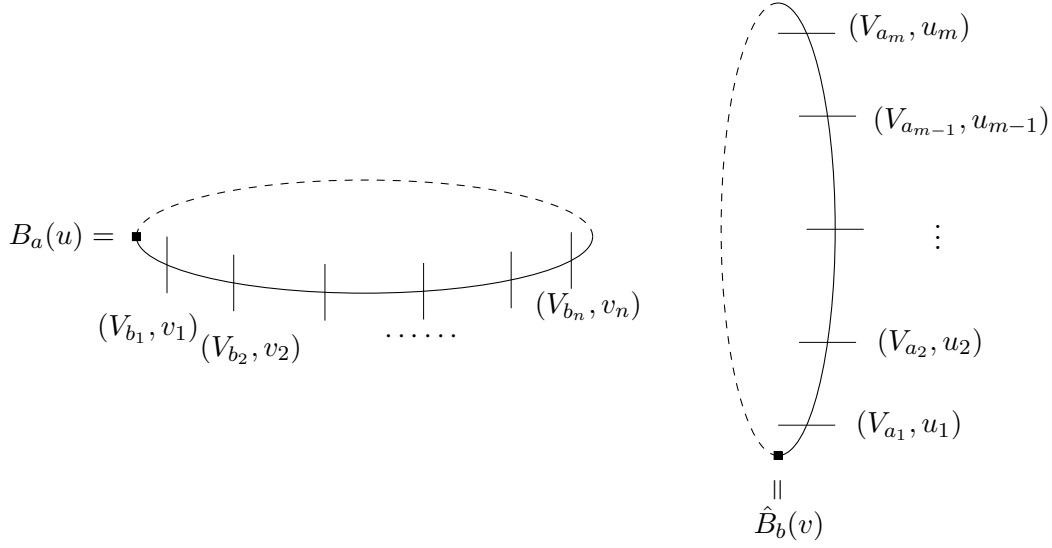
\begin{figure}[h]
\begin{subfigure}[c]{0.50\textwidth}
\centering
\begin{tikzpicture}[scale=0.75]
\fill (-4.2,0) node[left] {$B_a(u)=$ };
\draw [dashed] (4,0) arc[x radius=4, y radius=1, start angle=0, end angle=180];
\draw (4,0) arc[x radius=4, y radius=1, start angle=0, end angle=-180];
%\foreach \i in {1,2,3,4,5}
%\draw[fill=black]  ($(-30*\i:4 and 1)$) circle (0.08cm); %same as x radius and y radius, angle:\i
%\foreach \i in {1,2,3,4,5}
%\draw [line width=0.1mm,   black]  ($(-30*\i:4 and 1)-(0,0.5)$) --  ($(-30*\i:4 and 1)+(0,0.5)$) node [below] {};
%\node at ($(0,-0.8)+(-30*5:4 and 1)$) {$(V_{b_1},v_{1})$};
%\node at ($(0,-0.8-0.3)+(-30*4:4 and 1)$) {$(V_{b_2},v_{2})$};
%\node at ($(0,-0.8)+(-30*3:4 and 1)$) {$\cdots$};
%\node at ($(-0.1,-0.85-0.3)+(-30*2:4 and 1)$) {$(V_{b_{n-1}},v_{n-1})$};
%\node at ($(0.1,-0.8)+(-30*1:4 and 1)$) {$(V_{b_n},v_{n})$};
%\draw [black] plot [only marks, mark=square*] coordinates {(-4,0)};
\def\ang{25};
\foreach \i in {1,2,3,4,5,6}
\draw [line width=0.1mm,   black]  ($(-\ang*\i:4 and 1)-(0,0.5)$) --  ($(-\ang*\i:4 and 1)+(0,0.5)$) node [below] {};
\node at ($(-1.5,-0.8)+(-\ang*5:4 and 1)$) {$(V_{b_1},v_{1})$};
\node at ($(-1.3,-1.0)+(-\ang*4:4 and 1)$) {$(V_{b_2},v_{2})$};
\node at ($(0,-0.8)+(-\ang*3:4 and 1)$) {$\cdots$$\cdots$};
\node at ($(0.3,-0.8)+(-\ang*1:4 and 1)$) {$(V_{b_n},v_{n})$};
\draw [black] plot [only marks, mark=square*] coordinates {(-4,0)};
\end{tikzpicture}
%    \caption{}
  \end{subfigure}
  \quad
  \begin{subfigure}[c]{0.50\textwidth}
\centering
\begin{tikzpicture}[scale=0.75,rotate=90]
\fill (-4.6,-0.3) node[left] {\rotatebox{90}{$=$}};
\fill (-5.2,-1) node[left] {$\hat B_b(v)$};
\draw [dashed] (4,0) arc[x radius=4, y radius=1, start angle=0, end angle=180];
\draw (4,0) arc[x radius=4, y radius=1, start angle=0, end angle=-180];
%\foreach \i in {1,2,3,4,5}
%\draw[fill=black]  ($(-30*\i:4 and 1)$) circle (0.08cm); %same as x radius and y radius, angle:\i
\foreach \i in {1,2,3,4,5}
\draw [line width=0.1mm,   black]  ($(-30*\i:4 and 1)-(0,0.5)$) --  ($(-30*\i:4 and 1)+(0,0.5)$) node [below] {};
\node at ($(0,-0.8-1)+(-30*5:4 and 1)$) {$(V_{a_1},u_{1})$};
\node at ($(0,-0.8-1)+(-30*4:4 and 1)$) {$(V_{a_2},u_{2})$};
\node at ($(0,-0.8-1)+(-30*3:4 and 1)$) {$\vdots$};
\node at ($(-0.1,-0.85-1.5)+(-30*2:4 and 1)$) {$(V_{a_{m-1}},u_{m-1})$};
\node at ($(0.1,-0.8-1)+(-30*1:4 and 1)$) {$(V_{a_m},u_{m})$};
\draw [black] plot [only marks, mark=square*] coordinates {(-4,0)};
\end{tikzpicture}
%    \caption{}
  \end{subfigure}
\caption{\label{figBop} Modified operators $B(u)$ (left) and $\hat B(v)$ (right). The twists are represented by the filled squares.
}
\end{figure}

Similarly, we can consider the single column monodromy matrix defined by,
\ben\label{monoThat}
\Hat T_{b}(v|\bar u)=\overrightarrow{\prod_{i=1}^m}R_{a_ib}(u_i,v)\,,
\een
(see right panel of Figure \ref{mono}).
It satisfies the RTT relation,
\ben\label{RTThat}
R_{ba}(v_b,v_a) \Hat  T_a(v_a|\bar u)\Hat T_b(v_b|\bar u)=\Hat T_b(v_b|\bar u)\Hat T_a(v_a|\bar u)R_{ba}(v_b,v_a)\,,
\een
from which the additional family of mutually commuting operators
\ben\label{hBop}
\Hat B(v)=\tr_b(\Hat B_b \Hat T_b(v|\bar u))\,,
\een
with arbitrary twist $\Hat B_b$ can be obtained, by tracing over a given auxiliary space $V_b$,
(See a graphical representation in the right panel of Figure \ref{figBop}).
A second reformulation of the matrix of partition function follows from $\Hat T$, namely,
\ben
Z^{b_1\dots b_n}_{a_1\dots a_m}(\bar u|\bar v)=\overrightarrow{\prod_{i=1}^n}\Hat T_{b_i}(v_i|\bar u).
\een
and leads to the partition function in terms of the transfer matrix $\Hat B(u)$,
\ben\label{ZtrB}
Z_{mn}(\bar u|\bar v|B|\overline B)=\tr_{\bar a} \Big((\prod_{i=1}^m B_{a_i}) \Hat B(\bar v) \Big)\,.
\een

We have therefore two equivalent transfer matrix formulations (\ref{ZtrA}) and (\ref{ZtrB}) of the partition
function (\ref{Znm}). Depending on the nature of the twists $\{B,\hat B\}$, we have different
boundary conditions for the partition function (see Figure \ref{PartitionF}). If both twists do not have
rank one, we have a torus type of partition function. If only one of the twists has rank one, we have a cylinder and if both satisfies this condition we have the plane. If
one or both of the twists have rank one, the trace reduces to matrix elements of the matrices
$B(\bar u)$ or $\Hat B(\bar v)$ (see next (\ref{ZexpV})).

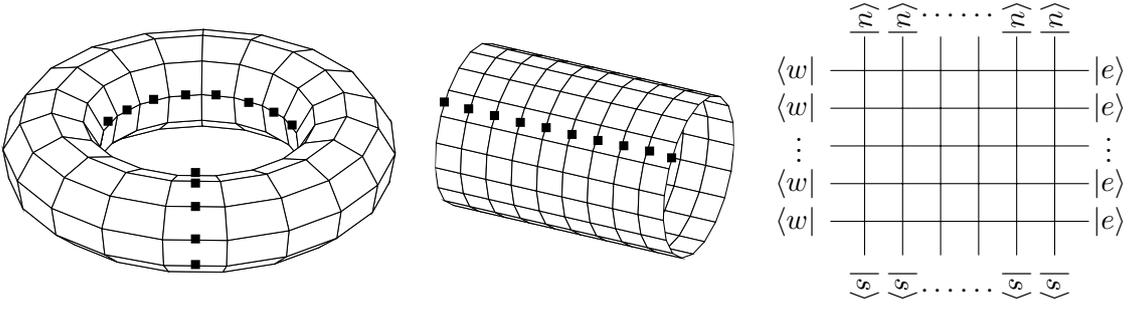
\begin{figure}[h!]
\begin{subfigure}[c]{0.30\textwidth}
\centering
\begin{tikzpicture}
	\begin{axis}[
		axis equal image,
		hide axis,
		z buffer = sort,
		scale = 1,
	]
		\addplot3[
		surf,
		fill=white,
		colormap/blackwhite,
		point meta=0,
		samples = 10,
		samples y =20,
		domain = 0:2*pi,
		domain y = 0:2*pi,
		thin,
		](
			{(3+sin(deg(\x)))*cos(deg(\y))},
			{(3+sin(deg(\x)))*sin(deg(\y))},
			{cos(deg(\x))}
		);
		\end{axis}
\begin{scope}[xshift = 2cm]
%second object
\draw [black] plot [only marks, mark size=1.5, mark=square*] coordinates {(1.35,2.07)};
\draw [black] plot [only marks, mark size=1.5, mark=square*] coordinates {(1.35,1.93)};
\draw [black] plot [only marks, mark size=1.5, mark=square*] coordinates {(1.35,1.62)};
\draw [black] plot [only marks, mark size=1.5, mark=square*] coordinates {(1.35,1.19)};
\draw [black] plot [only marks, mark size=1.5, mark=square*] coordinates {(1.35,0.85)};
\draw [black] plot [only marks, mark size=1.5, mark=square*] coordinates {(0.2,2.75)};
\draw [black] plot [only marks, mark size=1.5, mark=square*] coordinates {(0.45,2.9)};
\draw [black] plot [only marks, mark size=1.5, mark=square*] coordinates {(0.8,3.04)};
\draw [black] plot [only marks, mark size=1.5, mark=square*] coordinates {(1.22,3.10)};
\draw [black] plot [only marks, mark size=1.5, mark=square*] coordinates {(1.62,3.10)};
\draw [black] plot [only marks, mark size=1.5, mark=square*] coordinates {(2.05,3)};
\draw [black] plot [only marks, mark size=1.5, mark=square*] coordinates {(2.38,2.87)};
\draw [black] plot [only marks, mark size=1.5, mark=square*] coordinates {(2.62,2.7)};
\end{scope}
\end{tikzpicture}
%    \caption{}
  \end{subfigure}
 \quad\quad
  \begin{subfigure}[c]{.30\textwidth}
    \centering
\begin{tikzpicture}
	\begin{axis}[
		axis equal image,
		hide axis,
		z buffer = sort
	]
		\addplot3[
		surf,
		fill=white,
		colormap/blackwhite,
		point meta=0,
		samples = 10,
		samples y = 20,
		domain = 0:6,
		domain y = 0:2*pi,
		thin,
		](
			{\x},
			{2*sin(deg(\y))},
			{2*cos(deg(\y))}
		);
	\end{axis}
	\begin{scope}[xshift = 0cm]
%second object
\draw [black] plot [only marks, mark size=1.5, mark=square*] coordinates {(0.12,2.62)};
\draw [black] plot [only marks, mark size=1.5, mark=square*] coordinates {(0.44,2.53)};
\draw [black] plot [only marks, mark size=1.5, mark=square*] coordinates {(0.78,2.44)};
\draw [black] plot [only marks, mark size=1.5, mark=square*] coordinates {(1.12,2.35)};
\draw [black] plot [only marks, mark size=1.5, mark=square*] coordinates {(1.46,2.28)};
\draw [black] plot [only marks, mark size=1.5, mark=square*] coordinates {(1.8,2.19)};
\draw [black] plot [only marks, mark size=1.5, mark=square*] coordinates {(2.14,2.11)};
\draw [black] plot [only marks, mark size=1.5, mark=square*] coordinates {(2.48,2.04)};
\draw [black] plot [only marks, mark size=1.5, mark=square*] coordinates {(2.82,1.97)};
\draw [black] plot [only marks, mark size=1.5, mark=square*] coordinates {(3.11,1.88)};
\end{scope}
\end{tikzpicture}
   % \caption{}
  \end{subfigure}
  \begin{subfigure}[c]{.30\textwidth}
    \centering
\begin{tikzpicture}[scale=0.5]
\def\yy{0}
\draw (0.1,0.1+\yy) grid (6.9,5.9+\yy);
%\foreach \i in {1,...,5}
%    \foreach \j in {1,...,5}
%        \fill (\i,\j) circle (2pt);
\fill (0,1+\yy) node[left] {$\< w|$};
\fill (0,2+\yy) node[left] {$\< w|$};
\fill (-0.35,3.1+\yy) node[left] {$\vdots$};
\fill (0,4+\yy) node[left] {$\< w|$};
\fill (0,5+\yy) node[left] {$\< w|$};
%%%
\fill (8.15,1+\yy) node[left] {$|e\>$};
\fill (8.15,2+\yy) node[left] {$|e\>$};
\fill (8.15-0.35,3.1+\yy) node[left] {$\vdots$};
\fill (8.15,4+\yy) node[left] {$|e\>$};
\fill (8.15,5+\yy) node[left] {$|e\>$};
\fill (1,0+\yy) node[below] {\rotatebox{-90}{$|s\>$}};
\fill (2,0+\yy) node[below] {\rotatebox{-90}{$|s\>$}};
\fill (3.5,-0.15-0.25+\yy) node[below] {$\cdots\cdots$};
\fill (5,0+\yy) node[below] {\rotatebox{-90}{$|s\>$}};
\fill (6,0+\yy) node[below] {\rotatebox{-90}{$|s\>$}};
\fill (1,7.15+\yy) node[below] {\rotatebox{90}{$|n\>$}};
\fill (2,7.15+\yy) node[below] {\rotatebox{90}{$|n\>$}};
\fill (3.5,7.15-0.25+\yy) node[below] {$\cdots\cdots$};
\fill (5,7.15+\yy) node[below] {\rotatebox{90}{$|n\>$}};
\fill (6,7.15+\yy) node[below] {\rotatebox{90}{$|n\>$}};
\end{tikzpicture}
%    \caption{}
  \end{subfigure}
\caption{\label{PartitionF} Partition function in different geometries according to the nature
of the twist.
}
\end{figure}

Indeed, let us suppose that $\rank(B)=\rank(\Hat  B)=1$ and therefore that $\det(B)=\det(\Hat B)=0$.
We can then find the following bi-vector formulation,
\ben
&&B=|e\>\otimes \<w|=\sum_{ij}w_ie_jE_{ij},\quad \Hat  B=|n\>\otimes \<s|=\sum_{ij}n_is_jE_{ij}
\een
where
\ben\label{vecnot}
|x\>=\sum_{i} x_i |i\> , \quad
\<x|=\sum_{i} x_i  \<i|
\een
for arbitrary labels $x\in\{ w,e,n,s\}$ and we use the relation $|i\>\otimes \<j|=E_{ij}$ with $|i\>$ the vector with $1$ at the row $i$ and zero elsewhere, $\<j|$ is its dual and  $\<j|i\>=\delta_{ij}$.
It allows to rewrite operators
(\ref{Bop}) and (\ref{hBop}) as,
\ben\label{modB}
&&B(u)=\tr_a(B_aT_a(u|\bar v))=\null_a\<w|T_a(u|\bar v)|e\>_a=\sum_{ij}w_ie_jt_{ij}(u),\\
&&\Hat  B(v)=\tr_b(\Hat  B_b\Hat  T_b(v|\bar u))=\null_b\<s|\Hat  T_b(v|\bar u)|n\>_b=\sum_{ij}s_in_j\hat t_{ij}(v)\,.
\een
We recall that operators of this type arise in the modified Bethe ansatz \cite{BePi15},
and can be seen as a null transfer matrix as $\det(B)= \det(\hat B)=0$ \cite{BSV18}. They are represented
in Figure \ref{openB}.

\begin{figure}[h]
\begin{subfigure}[c]{0.50\textwidth}
\centering
    \begin{tikzpicture}[scale=0.75]
\draw (0.1,0.1) grid (6.9,1.9);
\fill (0,1) node[left] {$\<w|$};
\fill (8,1) node[left] {$|e\>$};
\fill (-1.0,1) node[left] {$B_a(u)=$ };
\fill (1,0) node[below] {$(V_{b_1},v_{1})$};
\fill (3.5,-0.2) node[below] {$\cdots$$\cdots$$\cdots$};
\fill (6,0) node[below] {$(V_{b_{n}},v_{n})$};
\end{tikzpicture}
%    \caption{}
  \end{subfigure}
  \quad
  \begin{subfigure}[c]{0.50\textwidth}
\centering
\begin{tikzpicture}[scale=0.75]
\draw (0.1,0.1) grid (1.9,5.9);
\fill (2.5,1) node[right] {$(V_{a_1},u_{1})$};
\fill (2.5,2) node[right] {$(V_{a_2},u_{2})$};
\fill (4-0.5,3) node[right] {$\vdots$};
\fill (2.5,4) node[right] {$(V_{a_{m-1}},u_{m-1})$};
\fill (2.5,5) node[right] {$(V_{a_m},u_{m})$};
\fill (1,0) node[below] {\rotatebox{-90}{$|s\>$}};
\fill (1,6.75) node[below] {\rotatebox{90}{$|n\>$}};
\fill (1,-1) node[below] {\rotatebox{90}{$=$}};
\fill (1,-2) node[below] {$\Hat B(v)$};
\end{tikzpicture}
%    \caption{}
  \end{subfigure}
\caption{\label{openB} Modified operators when $\rank(B)=\rank(\hat B)=1$.
}
\end{figure}
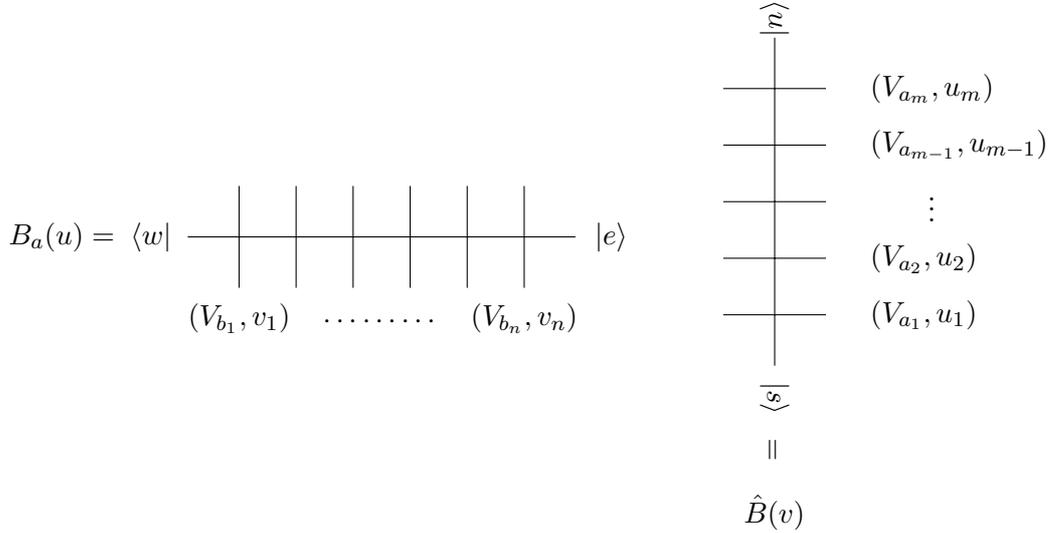

Note that we have the trace identity
\ben\label{bitr}
\tr_a(\hat B_a X_a)={}_a\<s|X_a|n\>_a
\een
for any matrix $X$ acting on vector space $V_a$.
Then, the general partition function can be rewritten into the forms of expectation values of product of operators,
\ben\label{ZexpV}
 Z_{mn}(\bar u|\bar v|B|\Hat  B)=\<S|B(\bar u)|N\>=\<W|\Hat  B(\bar v)|E\>
\een
with
\ben\label{compass}
\<W|= \<w|\otimes\dots \otimes \<w|, \quad
|E\>=|e\>\otimes\dots \otimes |e\>\\
\<S|= \<s|\otimes\dots \otimes \<s|, \quad
|N\>=|n\>\otimes\dots \otimes |n\>\,,
\een
which can be interpreted as a partition function for the six vertex model
with four arbitrary walls. The form (\ref{ZexpV}) is suitable
to be treated within the quantum inverse scattering method and linear algebra.
For the full open case both forms are possible and there is no quantization and no Bethe equations. We will see in the next section the rational six vertex case and its determinant representations.

\section{Rational six vertex model case: auxiliary operators and triangular representation theory}\label{sec:3}

All considerations so far are valid for vertex models
having an arbitrary number of states per edge. For concreteness, in the following we only consider
the symmetric (zero field) six vertex model, with $2$ states in all edges of the lattice
and therefore
 $V_{a_i}=V_{b_j}=\CC^2$.
The three possible nonzero Boltzmann weights,
called type $a$, $b$ and $c$, are represented in Figure \ref{F-abc-types}.
\begin{figure}[h!]
%\hspace{1cm}\includegraphics[height=25mm]{Vertex-W.eps}
\begin{picture}(500,70)
\put(30,0){%
\begin{picture}(70,70)
\put(30,10){\line(0,1){40}}
\put(10,30){\line(1,0){40}}
%\put(27.5,27){$\bullet$}
%
\put(33,10){$\scriptstyle j$}
\put(33,45){$\scriptstyle j$}
\put(10,33){$\scriptstyle j$}
\put(48,33){$\scriptstyle j$}
%qqq
\put(28,60){$v$}
\put(-5,28){$u$}
\put(60,28){$=\frac{u-v+c}{c}$}
\put(0,-10){$a$-type vertex}
\end{picture}
}
%%%%%%%%******************************************a-type
\put(180,0){%
\begin{picture}(70,70)
\put(30,10){\line(0,1){40}}
\put(10,30){\line(1,0){40}}
%\put(27.5,27){$\bullet$}
%
\put(33,10){$\scriptstyle j$}
\put(33,45){$\scriptstyle j$}
\put(10,33){$\scriptstyle k$}
\put(48,33){$\scriptstyle k$}
\put(28,60){$v$}
\put(-5,28){$u$}
\put(60,28){$=\frac{u-v}{c}$}
\put(0,-10){$b$-type vertex}
\end{picture}
}
%%%%%%%%******************************************b-type
\put(330,0){%
\begin{picture}(70,70)
\put(30,10){\line(0,1){40}}
\put(10,30){\line(1,0){40}}
%\put(27.5,27){$\bullet$}
%
\put(33,10){$\scriptstyle j$}
\put(33,45){$\scriptstyle k$}
\put(10,33){$\scriptstyle j$}
\put(48,33){$\scriptstyle k$}
\put(28,60){$v$}
\put(-5,28){$u$}
\put(60,28){$=1$}
\put(0,-10){$c$-type vertex}
\end{picture}
}
%kkk
\end{picture}
\vspace{0.1cm}
\caption{\label{F-abc-types} The nonzero Boltzmann weights $a$, $b$ and $c$ of the six vertex model.}
\end{figure}
The corresponding R-matrix is given by
\ben\label{6vertex}
R_{ab}(u-v)=\frac{u-v}{c}I_{ab}+P_{ab}.
\een
with $I_{ab}$ the identity operator and $P_{ab}$ the permutation operator on $\CC^2\otimes \CC^2$. It is
one of the simplest solution of the Yang-Baxter equation (\ref{YB}) and can be
used to define the Yangian $Y(gl_2)$, as we briefly recall in Appendix \ref{sec:yangian}.

Let us introduce additional operators $A(u)$ and $D(u)$. First consider the two zero determinant 2 by 2 matrices in bi-vector form, where the notations are adopted as in (\ref{vecnot}),
\ben
&&A=|a\>\otimes \<\tilde{a}|=\sum_{ij}\tilde a_ia_jE_{ij},\\
&&D=|d\>\otimes \<\tilde{d}|=\sum_{ij}\tilde d_id_jE_{ij}.
\een
They satisfy $\det(A)=\det(D)=0$  by definition. Then, we have the operators,
\ben\label{modAD}
A(u)=\tr_a(A_aT_a(u))=\<\tilde{a}|T_a(u)|a\>=\sum_{ij}a_{j}\tilde{a}_{i}t_{ij}(u),\\
D(u)=\tr_a(D_aT_a(u))=\<\tilde{d}|T_a(u)|d\>=\sum_{ij}d_{j}\tilde{d}_{i}t_{ij}(u).
\een
Imposing the relations
\ben\label{comconst}
B_aA_bP_{ab}=B_aA_b,\\
P_{ab}B_aD_b=B_aD_b,
\een
which are equivalent to the constraints $\tilde a_j=w_j$ and $ d_j=e_j$, where normalizations were absorbed in the free parameters $a_j$ and $\tilde d_j$, we obtain the following exchange relations,
\ben\label{modRTT}
A(u)B(v)=f(v,u)B(v)A(u)+g(u,v)B(u)A(v),\\
D(u)B(v)=f(u,v)B(v)D(u)+g(v,u)B(u)D(v).
\een
They directly {follow} from the definition of $\{A(u),D(u),B(u)\}$ as linear combinations of the
Yangian generators $t_{ij}(u)$ (see appendix \ref{sec:yangian}), and can also be calculated directly from (\ref{auxformRTT}) multiplying by $A_aB_b$, taking the traces over the spaces $V_a$ and $V_b$, and finally using relations (\ref{comconst}).

From linear combination of the results of  Theorem \ref{th1},
we can find the action of these operators on
the general states (\ref{compass}).
We find the following action of such operators on the general states (\ref{compass}),
\ben\label{rightaction}
&&A(u) |N\>=a_N \lambda_1(u)|N\>+c_N B(u) |N\>,\\
&&D(u) |N\>=d_N \lambda_2(u)|N\>+f_N B(u) |N\>,
\een
with
\ben
a_N={\frac{\<a|\sigma_y|e\>\<w|n\>}{\<n|\sigma_y|e\>}}\,,\quad
c_N={\frac{\<a|\sigma_y|n\>}{\<e|\sigma_y|n\>}}\,,\quad
d_N={\frac{\<e|\sigma_y|n\>\<w|\sigma_y|\tilde d\>}{\<w|n\>}}\,,\quad
f_N=\frac{\<\tilde d|n\>}{\<w|n\>}.
\een
and
\ben\label{lam12}
 \lambda_1(u_j)=h(u_j,\bar v)\,,\quad \lambda_2(u_j)=\frac{1}{g(u_j,\bar v)}\,.
\een
Note that we have the relations
\ben\label{rela1}
a_N+\tr(B)c_N=\tr(A), \quad
d_N+\tr(B)f_N=\tr(D).
\een

Similarly, the action of such operators on the general dual states (\ref{compass}) is given by,
\ben\label{leftaction}
&&\<S|A(u)  =a_S\lambda_2(u)  \<S|+c_S\<S|B(u),\nonumber \\
&&\<S|D(u)  =d_S\lambda_1(u)  \<S|+f_S\<S|B(u) ,
\een
with
\ben
a_S=\frac{\<a|\sigma_y|e\>\<s|\sigma_y|w\>}{\<s|e\>}\,,\quad
c_S=\frac{\<s|a\>}{\<s|e\>}\,,\quad
d_S=\frac{\<s|e\>\<\tilde d|\sigma_y|w\>}{\<s|\sigma_y|w\>}\,,\quad
f_S=\frac{\<s|\sigma_y|\tilde d\>}{\<s|\sigma_y|w\>}.
\een
and the relations
\ben\label{rela2}
 a_S+\tr(B)c_S=\tr(A)\,,\quad
 d_S+\tr(B)f_S=\tr(D).
\een

For further convenience, we introduce the parameter $\beta$ given by,
\ben
\beta=\frac{a_S}{a_N}=\frac{d_N}{d_S}
=\frac{\<w|\sigma_y|s\>\<e|\sigma_y|n\>}{\<e|s\>\<w|n\>}=
-\frac{\tr(B\sigma_y {\hat B}^t \sigma_y )}{\tr(B\hat B)}
\een
that also satisfies,
\ben
\tr(B\sigma_y {\hat B}^t \sigma_y )+\tr(B\hat B)=\tr(B)\tr({\hat B}).
\een

\section{Linear system and modified Izergin determinant formula}\label{sec:4}

In this section we derive a linear system for the partition function following the method
proposed in
\cite{BS19}. The first step is to note that the exchange relations (\ref{modRTT}) imply the following
multiple actions,
\ben\label{mult1}
A(u_i)B(\bar u_i)&=&\sum_{j=1}^{m+1}\frac{f(\bar u_j,u_j)}{h(u_i,u_j)}B(\bar u_j)A(u_j)\,,\\
D(u_i)B(\bar u_i)&=&\sum_{j=1}^{m+1}\frac{f( u_j,\bar u_j)}{h(u_j,u_i)}B(\bar u_j)D(u_j)\,,
\een
for a set $\bar u=\{u_1,\dots,u_{m+1}\}$ which includes an extra parameter $u_{m+1}$.

Next, we act with (\ref{mult1}) on the general states $\{|N\>,\<S|\}$. For convenience, consider the quantities,
\ben
F_A=\<S|(A(u_i)- c_S B(u_i))B(\bar u_i)|N\>\,,\\
F_D=\<S|(D(u_i)- f_S B(u_i))B(\bar u_i)|N\>,
\een
and
compute their actions to the left and to the right taking into account (\ref{leftaction},\ref{rightaction}).
We obtain,
\ben
\label{systemA1}
&& \frac{\tr(B)}{\chi}\sum_{j=1}^{m+1}\Big(-\beta\lambda_2(u_i)\delta_{ij} + \lambda_1(u_j)\frac{f(\bar u_{j},u_j)}{h(u_i,u_j)}\Big )Z_{mn}(\bar u_j|\bar v|B|\Hat B)=Z_{m+1n}(\bar u|\bar v|B|\Hat B)\,,\qquad \\
\label{systemD1}
&&\frac{\tr(B)}{\chi}\sum_{j=1}^{m+1}\Big(\lambda_1(u_i) \delta_{ij}-\beta\lambda_2(u_j)\frac{f( u_{j},\bar u_j)}{h(u_j,u_i)}\Big) Z_{mn}(\bar u_i|\bar v|B|\Hat B)=Z_{m+1n}(\bar u|\bar v|B|\Hat B)\,,\qquad
\een
where
\ben
\chi=1-\beta=\frac{\tr(\hat B)\tr(B)}{\tr(B \hat B)}.
\een
In the case $m=n$ the following theorem was proven in \cite{BSV18},
\begin{theorem}
\ben\label{Zred}
Z_{n+1n}(\bar u|\bar v|B|\Hat B)=\tr(B)\sum_{j=1}^{n+1} g( u_{j},\bar u_{j})\lambda_1( u_{j})\lambda_2(u_{j}) Z_{nn}(\bar u_{j}|\bar v|B|\Hat B)\,.
\een
\end{theorem}
Now, define a vector $\vec{X}=\left(X_1,\dots,X_{m+1}\right)^t$ with $X_j=Z_{nn}(\bar u_{j}|\bar v|B|\Hat B)$.
Using (\ref{Zred}) in (\ref{systemA1},\ref{systemD1}) for $m=n$, we obtain the following homogeneous linear systems,
\ben\label{sys}
L_A\vec{X}=0,\qquad
L_D\vec{X}=0,
\een
where the matrices $L_A,L_D$ with dimension $(n+1)\times (n+1)$ have entries given by
\ben
(L_A)_{ij}=-\beta \lambda_2(u_j)\delta_{ij}+g(u_j,\bar u_j) Y_A(u_j|\bar u_{i}),\\
(L_D)_{ij}=\lambda_1(u_j)\delta_{ij}-g(u_j,\bar u_j) Y_D(u_j|\bar u_{i}),
\een
with
\ben
 &&Y_A(u_j|\bar u_{i})=\tilde h(u_j,\bar u_i)\lambda_1(u_j)-\chi\lambda_1(u_{j})\lambda_2( u_{j}), \\
 &&Y_D(u_j|\bar u_{i})=\beta h( u_{j},\bar u_i)\lambda_2(u_j) +\chi\lambda_1(u_{j})\lambda_2( u_{j}),
\een
where we used (\ref{rela1},\ref{rela2}) for the constants.

Each homogeneous system in (\ref{sys}) has a nontrivial solution if
$\det(L_{A,D})=0$ which implies that $\rank(L_{A,D})\leq n$. To prove that the determinants
vanish, we define the nonsingular $(n+1)\times (n+1)$ matrix $W$ with entries,
\ben
W_{ij}=\frac{g(u_j,\bar u_j)}{g(u_j,\bar w_i)}\,,
\een
where a new set of arbitrary pairwise distinct parameters $\bar w=\{w_1,\dots, w_{ n+1}\}$, that will be
specified later, is introduced. The determinant of $W$ is given by
the ratio of determinants of the Vandermonde type,
\ben\label{WDelta}
\det(W)=\frac{\Delta(\bar w)}{\Delta(\bar u)}, \quad \Delta(\bar u)=(\prod_{i<j}g(u_i,u_j))^{-1}\,.
\een
Now, we define the products $\tilde L_A=WL_A$ and $\tilde L_D=WL_D$. Using the relations,
\ben
\sum_{k=1}^{n+1}W_{ik}=1, \quad \sum_{k=1}^{n+1}W_{ik}\tilde h(u_j,\bar u_{k})=\tilde h(u_j,\bar w_{i}), \quad \sum_{k=1}^{n+1}W_{ik}h(u_j,\bar u_{k})=h(u_j,\bar w_{i}),
\een
that can be proven using  contour integral for appropriate rational functions (see \cite{S22}),
we find,
\ben\label{afterTFA}
(\tilde L_A)_{ij}&=&g(u_j,\bar u_j)
\left(
\,-
\beta\frac{\lambda_2(u_j)}{g(u_j,\bar w_i)}+\lambda_1(u_j)\tilde h(u_j,\bar w_i)-\chi\lambda_1(u_j)
\lambda_2(u_j)
\right),\\ \label{afterTFD}
(\tilde L_D)_{ij}&=&g(u_j,\bar u_j)
\left(
\frac{\lambda_1(u_j)}{g(u_j,\bar w_i)}-\beta
\lambda_2(u_j) h(u_j,\bar w_i)-\chi\lambda_1(u_j)
\lambda_2(u_j)
\right).
\een %HERE
Let $i=n+1$. We set $w_j=v_j-c$ for $j\neq n+1$ in the $A$ system and  $w_j=v_j$ for $j\neq n+1$ in the  $D$ system. We also define $w_{n+1}=w$. Then, we find, respectively,
\ben \label{speafterTFD}
(\tilde L_A)_{n+1j}|_{w_j=v_j-c,w_{n+1}=w}=g(u_j,\bar u_j)\lambda_1(u_j)\lambda_2(u_j)(1-\beta-\chi)=0,\\
( \tilde L_D)_{n+1j}|_{w_j=v_j,w_{n+1}=w}=g(u_j,\bar u_j)\lambda_1(u_j)\lambda_2(u_j)(1-\beta-\chi)=0,
\een
from which it follows that $\det(L_A)={ \det(L_D)}=0$.
For the other rows $1\leq i\leq n$, we have
 \ben
(\tilde L_A)_{ij}=\tilde h(v_i,w) g(u_j,\bar u_j)\lambda_1(u_j)\lambda_2(u_j)\Big(-\beta g(u_j,v_i-c) +g(u_j,v_i)\Big),\\
( \tilde L_D)_{ij}= \frac{g(u_j,\bar u_j)}{g(v_i,w)}\lambda_1(u_j)\lambda_2(u_j)\Big(-\beta g(u_j,v_i-c)+g(u_j,v_i)\Big).
\een
 After renormalization with appropriate diagonal matrices,
 \ben
(\tilde T_A)_{ij}= \frac{\delta_{ij}}{ \tilde h(v_i,w)}\,, \quad
( \tilde T_D)_{ij}= g(v_i,w) \delta_{ij}\,,
\een
we find the same transformed linear system
for both $A$ and $D$ systems, namely,
\ben
(\tilde L)_{ij}=g(u_j,\bar u_j)\lambda_1(u_j)\lambda_2(u_j)
\Big(-\frac{\beta}{h(u_j,v_i)}+ g(u_j, v_i)\Big).
\een
Then, from Cramer's solution of the homogeneous linear system and  relations (\ref{lam12}),
it follows that,
\ben\label{Zdet}
Z_{nn}(\bar u|\bar v|B|\Hat B)=\phi(\bar v)\frac{\det_n(M)}{\det_n(C)} \label{Cramersol}\,.
\een
{ Here $\phi(\bar v)$ is
some symmetric function in $\bar v$, and}
\ben
M_{ij}=(\tilde L)^{(n+1)}_{ij}\frac{g(u_j,\bar v)}{g(u_j,\bar u_j)}=h(u_j,\bar v)\Big(-\beta \frac{1}{h(u_j,v_i)}+ g(u_j, v_i)\Big)
\een
are the entries of the reduced matrix $(\tilde L)^{(n+1)}_{ij}$ corresponding to the matrix
$\tilde L$ with the $n+1$ row and column removed. We also introduced a Cauchy matrix
with elements
\ben
C_{ij}=g(u_i,v_j)\,,
\een
that have the well known determinant,
\ben
\det_n(C)=g(\bar u,\bar v)\Delta(\bar u)\Delta'(\bar v).
\een
{ Here} $\Delta'(\bar u)=(\prod_{i>j}g(u_i,u_j))^{-1}$ which is similar to $\Delta(\bar u)$ in (\ref{WDelta}).

We finally need to fix the normalization in $\bar v$. This is done by
comparing the asymptotic limit $\bar u\rightarrow \infty$ from the
determinant formula (\ref{Zdet}) and from the definition of the partition
function (\ref{ZexpV}). The later follows simply from the asymptotic of the $B(u)$ operator
due to the R-matrix (\ref{6vertex}),
\ben\label{AsympB}
\lim_{ u \to \infty} B(u) &=& \left(\frac{u}{c}\right)^{n}\tr(B)+\cdots \,,
\een
which leads to,
\ben\label{Zlimit}
\lim_{\bar u \to \infty} \Big(Z_{nn}(\bar u|\bar v|B|\Hat B)/(\bar u/c^n)^n \Big)=(\tr(\hat B)\tr(B))^n+\cdots
\een
On the other hand, we have to take the asymptotic limit on the determinant form (\ref{Zdet}).
To do that we use the inverse of the Cauchy determinant, see {\it e.g.} \cite{S22}, given by,
\ben
C^{-1}_{kl}=g(u_l,v_k) \frac{g(\bar v_k,v_k)g( u_{ l},\bar u_{ l})}{g(\bar u,v_k)g( u_l,\bar v)},
\een
that satisfies the following summation rules,
\ben
\sum_{l=1}^nC^{-1}_{{ i}l}g(u_{ l},v_j)=\delta_{ij},
\een
and
\ben
\sum_{l=1}^nC^{-1}_{{ i}l}\frac{1}{h(u_l,v_j)}=\frac{g(v_i,\bar v_i)\tilde{h}(v_j,\bar v_j)}{g(\bar u,v_i)h(\bar u,v_j)h(v_i,v_j)}
\,.
\een
It follows that (\ref{Zdet}) can be put into the form,
\ben
Z_{nn}(\bar u|\bar v|B|\Hat B)= \frac{\phi(\bar v)}{g(\bar u,\bar v)}\det_n\Big(-\beta\delta_{ij}+ \frac{f(\bar u,v_i)f(v_i,\bar v_i)}{h(v_i,v_j)}\Big).
\een
Then, we can take the infinite limit for the set $\bar u$ and find
\ben
\lim_{\bar u \to \infty} Z/(\bar u/c^n)^n=\phi(\bar v) \det_n \Big( -\beta\delta_{ij}+\frac{f(v_j,\bar v_j)}{h(v_{i},v_{j})}\Big)+\cdots
\een
It follows from (\ref{oK-sumpart}) and Lemma 4.1 in \cite{BS18}, namely, from the identity
\ben
\sum f(\bar u_{\qo},\bar u_{\qt})=\left(\begin{array}{c}n \\ p \end{array}\right),
\een
where the sum is taken over partition of the set $\bar u$ into two sets $\{\bar u_{\qo},\bar u_{\qt}\}$ with fixed size  $\bar u_{\qo}=p$, $\bar u_{\qt}=n-p$. From Newton binomial formula it follows
\ben\label{Zlimit2}
\lim_{\bar u \to \infty} Z/(\bar u/c^n)^n=\phi(\bar v)\chi^n+\cdots\,.
\een
{ Comparing} (\ref{Zlimit2}) with (\ref{Zlimit}) we find,
\ben
\phi(\bar v) =\Big(\frac{\tr(\hat B)\tr(B)}{\chi}\Big)^n,
\een
that leads to
\ben\label{Znn}
Z_{nn}(\bar u|\bar v|B|\Hat B)=\frac{\tr(B)^n\tr(\hat B)^n}{\chi^n}\lambda_2(\bar u)K^{(\beta)}_{nn}(\bar u|\bar v),
\een
where $K^{(\beta)}_{nn}$ is the modified Izergin determinant (\ref{defKdef1}) for $m=n$, with $\lambda_2(\bar u)=(g(\bar u, \bar v))^{-1}$.

Finally, using the limits of the modified Izergin determinant in Proposition (\ref{SOP}), the formula
(\ref{Znn}) can be extended for values of $m\neq n$. For $m>n$, we use the limit (\ref{limKv})
to eliminate one $v_j$, while for $n>m$ we use the limit (\ref{limKu}) one $u_j$. Comparing with operator
 form (\ref{ZexpV}), we can find the general case $m,n$,
\ben\label{ZmodIKmn}
Z_{mn}(\bar u|\bar v|B|\Hat B)=\frac{\tr(B)^m \tr(\hat B)^n}{\chi^n}\lambda_2(\bar u)K^{(\beta)}_{mn}(\bar u|\bar v),
\een
that can be expressed
in two possible forms,
\be{defKdef1}
Z_{mn}(\bar u|\bar v|B|\Hat B)=\frac{\tr(B)^m \tr(\hat B)^n}{\chi^n} \lambda_2(\bar u)\det_n\left(-\beta\delta_{jk}+\frac{f(\bu,v_j)f(v_j,\bv_j)}{h(v_j,v_k)}\right),
\ee
and%
\be{defKdef2}
Z_{mn}(\bar u|\bar v|B|\Hat B)=\frac{\tr(B)^m \tr(\hat B)^n}{\chi^m}\lambda_2(\bar u)\det_m\left(\delta_{jk}f(u_j,\bv)-\beta \frac{f(u_j,\bu_j)}{h(u_j,u_k)}\right).
\ee
Furthermore, due to the properties of the modified Izergin determinant
in Proposition (\ref{Sform}),
we have the additional presentations as sum over partitions of the set $\bar v$,
\ben
Z_{mn}(\bar u|\bar v|B|\Hat B)=
\frac{\tr(B)^m \tr(\hat B)^n}{\chi^n} \lambda_2(\bar u) \sum_{\bv\Rightarrow\{\bv_{\so},\bv_{\st}\}}(-\beta)^{\#\bv_{\st}} f(\bu,\bv_{\so})f(\bv_{\so},\bv_{\st}).
\een
or over the set $\bar u$,
\ben
Z_{mn}(\bar u|\bar v|B|\Hat B)=
\frac{ \tr(B)^m \tr(\hat B)^n}{\chi^m}\lambda_2(\bar u) \sum_{\bu\Rightarrow\{\bu_{\so},\bu_{\st}\}}(-\beta)^{\#\bu_{\so}} f(\bu_{\st},\bv)f(\bu_{\so},\bu_{\st}).
\een

\section{Conclusion}\label{sec:conc}

Using the triangular representation theory of the modified operators (\ref{modB},\ref{modAD})
together with the exchange relations (\ref{mult1}), we found linear systems
that characterize, up to normalization, the partition function of the inhomogeneous rational six vertex model
on the rectangular lattice with two arbitrary twists $\{B,\hat B\}$ of rank $\rank(B)=\rank(\hat B)=1$.
Due to symmetry, only one of the eight generic boundary parameters associated
with the compass states (\ref{compass}) is free, that is, the partition function can be written
solely in terms of the parameter $\beta$, up to normalization by $\tr(B)^m\tr(\hat B)^n$.

The solution of the linear system and therefore the partition function on the rectangular lattice is given by a modified Izergin determinant (\ref{ZmodIKmn}).
For specific
values of the parameters,
one can recover known examples in the literature, for example, the partial domain wall partition function \cite{Foda2012} or the Izergin-Korepin determinant ($m=n$,  $s_1=0,s_2=1,e_1=0,e_2=1$ and $n_1=1,n_2=0,w_1=1,w_2=0$) \cite{Korepin1982,Izergin1987}.

For future problems, it is interesting to consider the linear system approach to the computation
of the partition function of the six vertex model with reflecting boundary conditions,
related to the reflection algebra \cite{Skly88}, which has been
extensively studied \cite{Tsuchiya1998,FiKi10,Lamers2018,FoZa16,PiPoVe17,Hietala2022}. Also in
this context, the linear system approach
may find applications in the connection between certain scalar products of modified Bethe states
and q-Racah polynomials \cite{BaPi23}.

Another intriguing problem is to study the trigonometric six vertex model (associated with
the XXZ chain) under anti-periodic boundary conditions.
The first step here would be to find appropriate modified operators and
off-shell relations, which are not yet known for this model, despite recent developments
(see
\cite{Zhang:2014fgb,NiTe15,ODBA} and references
therein).

It is also important to investigate higher rank vertex models, initially those based
on the $sl(n)$ algebra \cite{Reshetikhin1989}. In this case, the off-shell relations
and, therefore, the associated linear systems
are more intricate (see \cite{BPRS12,BPRS12b,LiSl18}).

%Finally, the homogeneous and thermodynamical limits need to be investigated
%generalizing  \cite{Korepin2000,Bleher2006,Tavares2015}.

\section*{Acknowledgements} S.B. { thanks} F. Levkovich-Maslyuk, R. Frassek, J. Jacobsen, V. Pasquier, D. Serban,  V. Terras for discussions. R.A.P. thanks R. Frassek for discussions and acknowledges support by the
German Research Council (DFG) via the Research Unit FOR
2316.

\appendix

\section{Basics about the Yangian  $Y(gl_2)$} \label{sec:yangian}
In this appendix, we review some basic facts about the Yangian of $gl_2$ denoted $Y(gl_2)$ (see
{\it e.g.} the
{ monograph} \cite{Chari,Molev2007} for details), as formulated in the quantum inverse scattering method. Let us introduce the monodromy matrix
\ben\label{MonoT}
T(u)=\left(\begin{array}{ll} t_{11}(u)& t_{12}(u)\\ t_{21}(u)&t_{22}(u)\end{array}\right),
\een
whose elements are given by the formal series in $u$,
\ben
 t_{ij}(u)=\delta_{ij}+\sum_{r=1}^\infty t_{ij}^{(r)}u^{-r},
\een
where $t_{ij}^{(r)}$ are the generators of the Yangian $Y(gl_2)$ subject to
the defining RTT relations
\ben\label{RTTa}
R_{ab}(u-v)T_a(u)T_b(v)= T_b(v)T_a(u)R_{ab}(u-v),
\een
where
$R$ is
the six vertex R-matrix (\ref{6vertex}). They encode
exchange relations for the elements $t_{ij}(u)$, for instance,
\ben
t_{11}(u)t_{12}(v)&=&f(v,u)t_{12}(v)t_{11}(u)+g(u,v)t_{12}(u)t_{11}(v),\\
t_{22}(u)t_{12}(v)&=&f(u,v)t_{12}(v)t_{22}(u)+g(v,u)t_{12}(u)t_{22}(v),\\
t_{12}(u)t_{12}(v)&=&t_{12}(v)t_{12}(u),
\een
where the rational functions $f,g$ in the formal variables $u,v$ are given by (\ref{funcs}).

Let us note that (\ref{RTTa}) can be rewritten as
\ben\label{auxformRTT}
[T_a(u),T_b(v)]= g(v,u)P_{ab}(T_a(u)T_b(v)-T_a(v)T_b(u)).
\een

We consider finite dimensional representations of the Yangian
\cite{Tarasov1985} from (\ref{monoT}). Highest and lowest weight representations can be respectively constructed
from vectors $|0\>$ and $|\hat 0\>$ with actions,
\ben\label{raction}
T_a(u)|0\>=\left(\begin{array}{cc}
       \lambda_1(u) & t_{12}(u)\\
      0 & \lambda_2(u)      \end{array}\right)_a|0\>,
\een
\ben\label{ractionh}
T_a(u)|\hat 0\>=\left(\begin{array}{cc}
       \lambda_2(u) & 0\\
       t_{21}(u) & \lambda_1(u)      \end{array}\right)_a|\hat 0\>.
\een
The dual analogs $\<0|$ and $\<\hat 0|$ have actions,
\ben\label{laction}
\<0|T_a(u)=\<0|\left(\begin{array}{cc}
       \lambda_1(u) & 0\\
      t_{21}(u) & \lambda_2(u)      \end{array}\right)_a,
\een
\ben\label{lactionh}
\<\hat 0|T_a(u)=\<\hat 0|\left(\begin{array}{cc}
       \lambda_2(u) & t_{12}(u)\\
      0 & \lambda_1(u)      \end{array}\right)_a.
\een

In particular, for the $6$ vertex case, we have the highest and lowest vectors,
\ben
|0\>&=&\otimes_{j=1}^n\left(\begin{array}{c}  1 \\0\end{array}\right),\quad
\<0|=\otimes_{j=1}^n\left(\begin{array}{cc}  1, & 0\end{array}\right),
\\
|\hat 0\>&=&\otimes_{j=1}^n\left(\begin{array}{c}  0 \\1\end{array}\right),
\quad
\<\hat 0|=\otimes_{j=1}^n\left(\begin{array}{cc}  0, & 1\end{array}\right),
\een
and the weight functions $\lambda_i(u)$ are given by (\ref{lam12}).

The following theorem provides the modified version of the actions
(\ref{raction},\ref{ractionh},\ref{laction},\ref{lactionh}).

\begin{theorem}\label{th1}
The action of the entries of the monodromy matrix on an arbitrary vector
\ben
|X\>=|x\>\otimes\dots \otimes |x\>, \quad \text{with} \quad |x\>=\Big(\begin{array}{c}  x_1 \\x_2\end{array}\Big),
\een
with $x_1\neq 0$
are given by
\ben
&&t_{11}(u) |X\>=\lambda_1(u) |X\>-\frac{x_2}{x_1} t_{12}(u) |X\>,\\
&&t_{22}(u) |X\>=\lambda_2(u) |X\>+\frac{x_2}{x_1} t_{12}(u) |X\>,\\
&&t_{21}(u) |X\>=\frac{x_2}{x_1}(\lambda_1(u)-\lambda_2(u) )|X\> -\Big(\frac{x_2}{x_1}\Big)^2t_{12}(u) |X\>,
\een
and with $x_2\neq 0$
are given by
\ben
&&t_{11}(u) |X\>=\lambda_2(u) |X\>+\frac{x_1}{x_2} t_{21}(u) |X\>,\\
&&t_{22}(u) |X\>=\lambda_1(u) |X\>-\frac{x_1}{x_2} t_{21}(u) |X\>,\\
&&t_{{12}}(u) |X\>=\frac{x_1}{x_2}(\lambda_{ 1}(u)-\lambda_{ 2}(u) )|X\>{ -}\Big(\frac{x_1}{x_2}\Big)^2t_{{21}}(u) |X\>.
\een
Similarly, we can find the action on the dual vector
\ben
\<X|= \<x|\otimes\dots \otimes \<x|,  \quad \text{with} \quad \<x|=(\begin{array}{cc} x_1,& x_2\end{array}),
\een
with $x_1\neq 0$
\ben
&&\<X|t_{11}(u) =\lambda_1(u) \<X|-\frac{x_2}{x_1}\<X| t_{21}(u) ,\\
&&\<X|t_{22}(u) =\lambda_2(u) \<X|+\frac{x_2}{x_1} \<X|t_{21}(u), \\
&&\<X|t_{12}(u) =\frac{x_2}{x_1}(\lambda_1(u)-\lambda_2(u) )\<X|{-}\Big(\frac{x_2}{x_1}\Big)^2\<X|t_{21}(u),
\een
with $x_2\neq 0$
\ben
&&\<X|t_{11}(u) =\lambda_2(u) \<X|{ +}\frac{x_1}{x_2}\<X| t_{12}(u), \\
&&\<X|t_{22}(u) =\lambda_1(u) \<X|{ -}\frac{x_1}{x_2} \<X|t_{12}(u) ,\\
&&\<X|t_{21}(u) =\frac{x_1}{x_2}(\lambda_{ 1}(u)-\lambda_{2}(u) )\<X|-\Big(\frac{x_1}{x_2}\Big)^2\<X|t_{12}(u).
\een
\end{theorem}
\begin{proof}
These relations can be proven in the following way. Let us introduce an invertible matrix $X$ such that $X|0\>= |x\>$ with the form,
\ben
X=\left(\begin{array}{ll} x_1 & \alpha \\ x_2 & \beta\end{array} \right).
\een

Then due to the $gl(2)$ invariance of the Yangian we have $[T_a(u),X_a\prod_{i=1}^nX_{b_i}]=0$ and it follows that,
\ben
T_a(u)|X\>&=&T_a(u)\big(\prod_{i=1}^nX_{b_i}\big)|0\>=T_a(u)X_a\big(\prod_{i=1}^nX_{b_i}\big)(X_a)^{-1}|0\>\\
&=&X_a\big(\prod_{i=1}^nX_{b_i}\big)T_a(u)|0\>(X_a)^{-1}.
\een
From the representation theory of the Yangian (\ref{raction},\ref{ractionh}), we have,
\ben
T_a(u)|0\>=\left(\begin{array}{cc}
       \lambda_1(u) &0\\
      0 & \lambda_2(u)      \end{array}\right)_a|0\>+\left(\begin{array}{cc}
       0&1\\
      0 &0    \end{array}\right)_a t_{12}(u)|0\>.
\een
It follows that,
\ben
T_a(u)|X\>&=&\Lambda_a(u)|X\>+E_a \big(\prod_{i=1}^NX_{b_i}\big)t_{12}(u)|0\>,
\een
with
\ben
\det(X)\Lambda_a(u)&=&X_a\left(\begin{array}{cc}
       \lambda_1(u) &0\\
      0 & \lambda_2(u)      \end{array}\right)_a(X_a)^{-1}\\
      &=&\left(\begin{array}{cc}
     x_1 \beta  \lambda_1(u)-  x_2 \alpha \lambda_2(u)  &  { -}x_{ 1} \alpha  \big( \lambda_1(u)- \lambda_2(u) \big )\\
    x_{2} \beta \big( \lambda_1(u)- \lambda_2(u)  \big) &  x_1 \beta  \lambda_2(u)-  x_2 \alpha \lambda_1(u)    \end{array}\right)_a,
\een
and
\ben
\det(X)E_a&=&X_a\left(\begin{array}{cc}
       0&1\\
      0 &0    \end{array}\right)_a(X_a)^{-1}=\left(\begin{array}{cc}
    - x_1  x_2   &  (x_1)^2 \\
  -(x_2)^2  &  x_1 x_2  \end{array}\right)_a\,.
\een
Then after some linear algebra we express $\big(\prod_{i=1}^NX_{b_i}\big)t_{12}(u)|0\>$ in terms of $t_{12}(u)|X\>$ and $|X\>$ we find the desired actions. The left actions can be proven similarly and
we omit it here.

\end{proof}

\section{Modified Izergin determinant}\label{sec:mi}

We recall some basic properties of the modified Izergin determinant (see \cite{BSV18a,BS19a} for more details).
The modified Izergin determinant can be {defined as follows.}
\begin{Def}
Let $\bu=\{u_1,\dots,u_m\}$,  $\bv=\{v_1,\dots,v_n\}$ and $z$ be a complex number.
Then  the modified Izergin determinant $K_{mn}^{(z)}(\bu|\bv)$
is defined by
\be{defKdef1}
K_{mn}^{(z)}(\bu|\bv)=\det_n\left(-z\delta_{jk}+\frac{f(\bu,v_j)f(v_j,\bv_j)}{h(v_j,v_k)}\right).
\ee
Alternatively the modified Izergin determinant can be presented as
\be{defKdef2}
K_{mn}^{(z)}(\bu|\bv)=(1-z)^{n-m}\det_m\left(\delta_{jk}f(u_j,\bv)-z\frac{f(u_j,\bu_j)}{h(u_j,u_k)}\right).
\ee
\end{Def}
The proof of the equivalence of representations \eqref{defKdef1} and \eqref{defKdef2} can be found in \cite{GZZ}.

In the particular case $z=1$ and $\#\bu=\#\bv=n$, the modified Izergin determinant
turns into the ordinary Izergin determinant, that we traditionally denote by $K_{n}(\bu|\bv)$,
\be{ModI-OrdI}
K_{nn}^{(1)}(\bu|\bv)=K_{n}(\bu|\bv).
\ee
It also follows
from \eqref{defKdef2} that
\be{K10}
K_{mn}^{(1)}(\bu|\bv)=0,\qquad \text{for} \qquad m<n.
\ee

Additional properties of the modified Izergin determinant needed here are given in the following propositions.

\begin{prop}\label{SOP}
\vspace{1mm}
We have the limits,
\be{limKv}
\lim_{v_j\to \infty}K^{(z)}_{mn}(\bu|\bv)=(1-z)K^{(z)}_{mn-1}(\bu|\bv_j),
\ee
\be{limKu}
\lim_{u_j\to \infty}K^{(z)}_{mn}(\bu|\bv)=K^{(z)}_{m-1n}(\bu_j|\bv).
\ee
\end{prop}
\begin{prop}\label{Sform}
We have the sum formulation,
\be{K-sumpart}
K_{mn}^{(z)}(\bu|\bv)=\sum_{\bv\Rightarrow\{\bv_{\so},\bv_{\st}\}}(-z)^{\#\bv_{\st}} f(\bu,\bv_{\so})f(\bv_{\so},\bv_{\st}),
\ee
where the sum is taken over all partitions $\bv\Rightarrow\{\bv_{\so},\bv_{\st}\}$, and
\be{oK-sumpart}
\begin{aligned}
&K_{m,n}^{(z)}(\bu|\bv)=(1-z)^{n-m}\sum_{\bu\Rightarrow\{\bu_{\so},\bu_{\st}\}}(-z)^{\#\bu_{\so}} f(\bu_{\st},\bv)f(\bu_{\so},\bu_{\st}),
\end{aligned}
\ee
where the sum is taken over all partitions $\bu\Rightarrow\{\bu_{\so},\bu_{\st}\}$.
\end{prop}

\printbibliography

%%%%%%%%%%%%%%%%%%%%%%%%%%%%%%%%%%%%%%%%%%%%%%%%%%%%%%%%%%
%%%%%%%%%%%%%%%%%%%%%%%%%%%%%%%%%%%%%%%%%%%%%%%%%%%%%%%%%%
%%%%%%%%%%%%%%%%%%%%%%%%%%%%%%%%%%%%%%%%%%%%%%%%%%%%%%%%%%

\end{document}